\documentclass[10pt,journal,compsoc]{IEEEtran}
\ifCLASSOPTIONcompsoc
  \usepackage[nocompress]{cite}
\else
  \usepackage{cite}
\fi
\ifCLASSINFOpdf
  \usepackage[pdftex]{graphicx}
\else
\fi
\usepackage{amsmath,amssymb,amsfonts}
\usepackage[linesnumbered, ruled,vlined]{algorithm2e} 
\usepackage{xcolor}
\usepackage{algpseudocode} 
\usepackage{graphicx}
\usepackage{textcomp}
\def\BibTeX{{\rm B\kern-.05em{\sc i\kern-.025em b}\kern-.08em
    T\kern-.1667em\lower.7ex\hbox{E}\kern-.125emX}}

\usepackage{textcomp}
\usepackage{caption}
\usepackage{subcaption}
\usepackage{amscd}
\usepackage{amsthm,enumerate}

\newtheorem{theorem}{Theorem}[section]
\newtheorem{lemma}{Lemma}[section] 

\theoremstyle{definition}
\newtheorem{definition}{Definition}[section]

\usepackage{pgfplots}
\usepackage{pgfplotstable}
\pgfplotsset{compat=1.7}
\usepackage{tikz}

\begin{document}
%
\title{Sharding-Based Proof-of-Stake Blockchain Protocols:
Security Analysis}
%
%
%
%

\author{Abdelatif~Hafid,~\IEEEmembership{}
        Abdelhakim~Senhaji~Hafid,~\IEEEmembership{}
        Adil~Senhaji,~\IEEEmembership{}
\IEEEcompsocitemizethanks{\IEEEcompsocthanksitem A. Hafid and A. S. Hafid are with Montreal Blockchain Laboratory (mbl), Department of Computer Science and Operational Research, University of Montreal, Montreal, QC H3T 1J4, Canada.\protect\\
E-mail: abdelatif.hafid@umontreal.ca
\IEEEcompsocthanksitem A. Senhaji  is a Managing Director of technology at Mizuho Securities, USA.
}

\thanks{}}

%
%

\markboth{}%
{Hafid \MakeLowercase{\textit{et al.}}: Sharded Blockchain} 
%



\IEEEtitleabstractindextext{%
\begin{abstract}
Blockchain technology has been gaining great interest from a variety of sectors, including healthcare, supply chain and cryptocurrencies. However, Blockchain suffers from its limited ability to scale (i.e. low throughput and high latency). Several solutions have been appeared to tackle this issue. In particular, sharding proved that it is one of the most promising solutions to Blockchain scalability. Sharding can be divided into two major categories: (1) Sharding-based Proof-of-Work (PoW) Blockchain protocols, and (2) Sharding-based Proof-of-Stake (PoS) Blockchain protocols. The two categories  achieve a good performances (i.e. good throughput with a reasonable latency), but raise security issues. This article attends that analyze the security of the second category. More specifically, we compute the probability of committing a faulty block and measure the security by computing the number of years to fail. Finally, to show the effectiveness of the proposed model, we conduct a numerical analysis and evaluate the results obtained.
\end{abstract}
}

\maketitle

\IEEEdisplaynontitleabstractindextext

%
\IEEEpeerreviewmaketitle
\IEEEraisesectionheading{\section{Introduction}\label{sec:introduction}} With the raise of Bitcoin \cite{nakamoto2008bitcoin}, Blockchain has attracted significant attention and extensive research. More specifically, it has been adopted by numerous and several applications, such as healthcare \cite{HC1}, supply chain \cite{SC1}, and Internet-of-Things \cite{IoT1}. However, Blockchain's capacity to scale is very limited \cite{hafid2020scaling}. For example, in the case of cryptocurrencies, Bitcoin \cite{nakamoto2008bitcoin} handles between 3-7 transactions per second (tps), which is very limited compared to traditional payment systems (e.g. PayPal \cite{paypal}). To deal with this issues, many solutions appeared \cite{hafid2020scaling}. In particular, sharding is emerged as a promising solution \cite{hafid2020scaling}. Sharding consists of partitioning the network into shanks, called shards; all shards work in parallel to enhance the performance of the network and then mitigate the scalability issues. More specifically, each shard handles a sub-set of transactions instead of the entire network handles all the transactions. However, the security of sharding-based blockchain protocols is emerging as a challenging issue. More specifically, in sharding-based blockchain environment, it is easy for a malicious user to conquer and attack a single shard compared to the whole network. This attack is well-known as a shard takeovers attack \cite{hafid2021tractable}.

Malicious nodes (e.g. Sybil nodes) are increasingly appear with the growing and the spread of Blockchain technology. In contract, several consensus mechanisms appeared to deal with these Byzantine nodes, including, PoW, PoS, , and practical Byzantine Fault Tolerance (pBFT).

pBFT is an algorithm that tolerates Byzantine fault \cite{castro1999practical}; this algorithm belongs to the BFT class. Byzantine Fault Tolerance (BFT) is the feature of a distributed network to reach consensus (agreement on the same value) even when some of nodes in the network fail to respond or respond with incorrect information. Leslie Lamport \cite{lamport2019byzantine} proved that if we have $3m+1$ correctly working processors, a consensus (agreement on same value) can be reached if at most $m$ processors are faulty, which means strictly more than two-thirds of the total number of processors should be honest.

PoS is an alternative consensus mechanism of PoW. It consists of selecting validators in proportion to their number of coins. Validators are responsible of adding new blocks to Blockchain. 

By going through the literature \cite{wang2019sok, hafid2020scaling}, we can classify sharding-based Blockchain protocols into two major categories: sharding-based PoW and sharding-based PoS Blockchain protocols.

Recently, Hafid et al. \cite{hafid2019new, hafid2021tractable, hafid2020novel} proposed many approaches and models to analyze the security of the first category. However, there is a lack of methods in the literature to analyze the second category. In this paper, we focus on the second category. Precisely, we provide a probabilistic model to analyze the security of this kind of sharding-based Blockchain protocols by computing the probability of committing a faulty block. Furthermore, based on these probabilities, we calculate the number of years to fail for the purpose of quantifying and measuring the security of the network.  

The rest of the paper is structured as follows. Section \ref{sec: Mathematical Model} presents an overview of sharding-based PoS Blockchain protocols and presents the proposed probabilistic model. Section \ref{sec: Results and Evaluation} presents numerical results and evaluates the proposed model. In the end, we conclude this paper in Section \ref{sec: Conclusion}.
\section{Mathematical Model} \label{sec: Mathematical Model}
The objective of this section is to propose a probabilistic model to analyze the security of sharding-based PoS Blockchain protocols.

\subsection{Abbreviations and Definitions} \label{AA}
In what follows, we provide abbreviations and definitions required for the rest of the paper. Table $\ref{table:1}$ shows the list of symbols and variables that are used to describe the proposed model. 

\begin{table}[ht]
\caption{Notations \& Symbols}
\label{table:1}
\setlength{\tabcolsep}{3pt}
\begin{tabular}{|p{65pt}|p{175pt}|}
\hline
Notation & 
Description \\
\hline
$N$&
Number of users\\
$n$ &
Size of a shard's committee\\
$n'$ &
Size of the beacon's committee\\
$H$ &
Number of honest validators in a shard\\
$M$ &
Number of malicious validators in a shard\\
$V$ &
Number of validators in a shard ($V = H + M$) \\
$\zeta$ &
Number of committees \\
$X$ &
Random variable that computes the number of malicious nodes in the committee of a shard \\
$H'$ &
Number of honest validators in the beacon chain\\
$M'$ &
Number of malicious validators in the beacon chain\\
$V'$ &
Number of validators in the beacon chain ($V' = H' + M' $) \\
$X'$ &
Random variable that computes the number of malicious nodes in the committee of the beacon chain \\
$r$ &
Committee/beacon resiliency\\
$R$ &
Shard resiliency\\
$R'$ &
Beacon resiliency\\
$\mathcal{P}$ &
Probability of conquering the protocol \\
$P$ &
Probability of a shard to commit a faulty block \\
$P^{'}$ &
Probability of the beacon chain to commit a faulty block \\
$P^{''}$ &
Probability of all shards committing a faulty block\\
$p_m$ &
Percentage of malicious validators in a shard chain as well as in the beacon chain\\
$Y_f$ &
Years to fail\\
\hline
\end{tabular}
\end{table}

\begin{definition}[Shard's committee]
Shard's committee is a subset of validators selected randomly from the set of validators that are decided to stake for the shard chains.
\end{definition}

\begin{definition}[Beacon's committee]
Beacon's committee is a subset of validators selected randomly from the set of validators that are decided to stake for the beacon chain.
\end{definition}

\subsection{Probabilistic Model} \label{PP}
In this section, we present the proposed probabilistic model. Figure \ref{fig:sharding_scheme} shows a network of sharding-based PoS Blockchain protocol. The network contains a single beacon chain and $\zeta$ shard chains. Shard chains forge/produce blocks in parallel. All shard chains are synchronized by beacon chain. More specifically, each shard has its own committee, which is randomly assigned by the beacon chain. And each shard committee processes transactions belonging to it. When a shard block is created, the beacon committee verify the block, if valid, it adds the block header to the beacon chain. Otherwise, it drops it and sends the proof to other shards to vote for the purpose of slashing the misbehaving shard committee. Furthermore, in each epoch (a period of time), beacon chain shuffles committees to ensure the security. For Incognito \cite{incognito}, when a new random number is generated, the beacon chain shuffles committees; this means that, one epoch for Incognito is corresponding to generating a new random number. This number is generated periodically in a round-robin fashion \cite{incognito}, \cite{rasmussen2008round}.

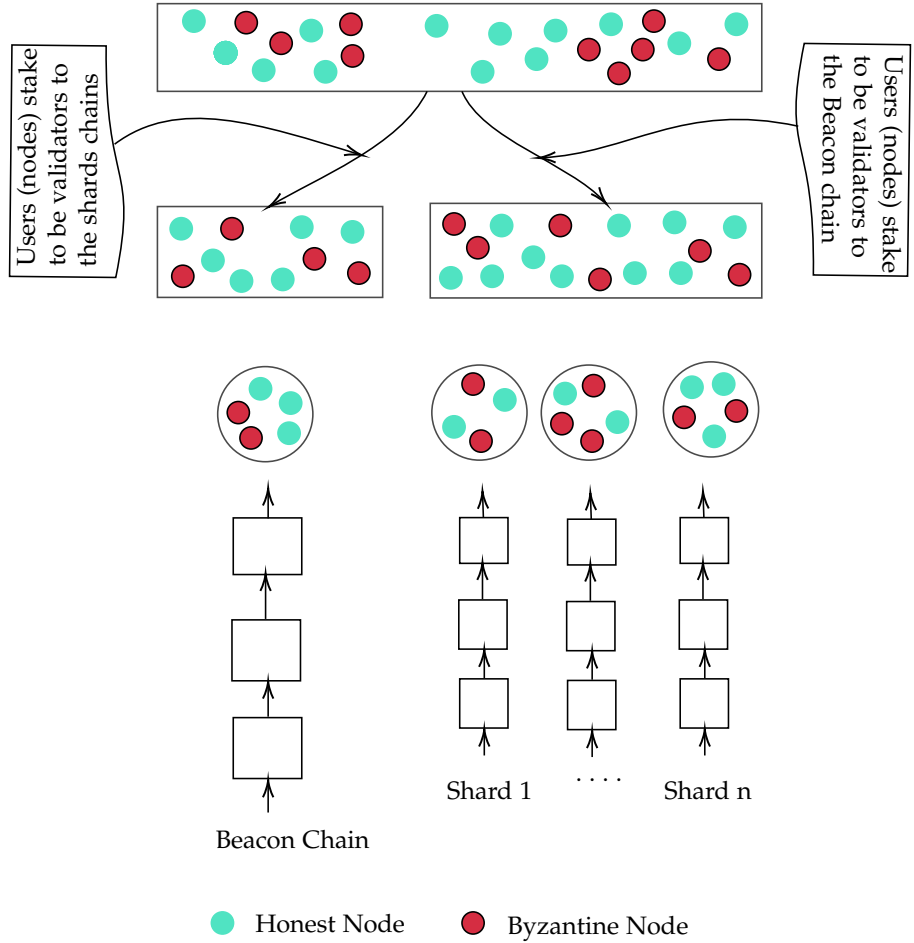
\begin{figure*}[ht]
\centering
\tikzset{every picture/.style={line width= 0.6pt}} 
\begin{tikzpicture}[x= 0.6pt,y= 0.6pt,yscale =-1, xscale = 1]
\draw  [color={rgb, 255:red, 74; green, 74; blue, 74 }  ,draw opacity=1 ] (130,13) -- (511,13) -- (511,68.33) -- (130,68.33) -- cycle ;
\draw  [color={rgb, 255:red, 74; green, 74; blue, 74 }  ,draw opacity=1 ] (130,141) -- (272,141) -- (272,198.33) -- (130,198.33) -- cycle ;
\draw  [color={rgb, 255:red, 74; green, 74; blue, 74 }  ,draw opacity=1 ] (302,139) -- (511,139) -- (511,198.67) -- (302,198.67) -- cycle ;
\draw  [color={rgb, 255:red, 80; green, 227; blue, 194 }  ,draw opacity=1 ][fill={rgb, 255:red, 80; green, 227; blue, 194 }  ,fill opacity=1 ] (146,24) .. controls (146,20.13) and (149.13,17) .. (153,17) .. controls (156.87,17) and (160,20.13) .. (160,24) .. controls (160,27.87) and (156.87,31) .. (153,31) .. controls (149.13,31) and (146,27.87) .. (146,24) -- cycle ;
\draw  [color={rgb, 255:red, 0; green, 0; blue, 0 }  ,draw opacity=1 ][fill={rgb, 255:red, 213; green, 45; blue, 66 }  ,fill opacity=1 ] (179,25) .. controls (179,21.13) and (182.13,18) .. (186,18) .. controls (189.87,18) and (193,21.13) .. (193,25) .. controls (193,28.87) and (189.87,32) .. (186,32) .. controls (182.13,32) and (179,28.87) .. (179,25) -- cycle ;
\draw  [color={rgb, 255:red, 80; green, 227; blue, 194 }  ,draw opacity=1 ][fill={rgb, 255:red, 80; green, 227; blue, 194 }  ,fill opacity=1 ] (166,44) .. controls (166,40.13) and (169.13,37) .. (173,37) .. controls (176.87,37) and (180,40.13) .. (180,44) .. controls (180,47.87) and (176.87,51) .. (173,51) .. controls (169.13,51) and (166,47.87) .. (166,44) -- cycle ;
\draw  [color={rgb, 255:red, 80; green, 227; blue, 194 }  ,draw opacity=1 ][fill={rgb, 255:red, 80; green, 227; blue, 194 }  ,fill opacity=1 ] (166,44) .. controls (166,40.13) and (169.13,37) .. (173,37) .. controls (176.87,37) and (180,40.13) .. (180,44) .. controls (180,47.87) and (176.87,51) .. (173,51) .. controls (169.13,51) and (166,47.87) .. (166,44) -- cycle ;
\draw  [color={rgb, 255:red, 80; green, 227; blue, 194 }  ,draw opacity=1 ][fill={rgb, 255:red, 80; green, 227; blue, 194 }  ,fill opacity=1 ] (166,44) .. controls (166,40.13) and (169.13,37) .. (173,37) .. controls (176.87,37) and (180,40.13) .. (180,44) .. controls (180,47.87) and (176.87,51) .. (173,51) .. controls (169.13,51) and (166,47.87) .. (166,44) -- cycle ;
\draw  [color={rgb, 255:red, 80; green, 227; blue, 194 }  ,draw opacity=1 ][fill={rgb, 255:red, 80; green, 227; blue, 194 }  ,fill opacity=1 ] (166,44) .. controls (166,40.13) and (169.13,37) .. (173,37) .. controls (176.87,37) and (180,40.13) .. (180,44) .. controls (180,47.87) and (176.87,51) .. (173,51) .. controls (169.13,51) and (166,47.87) .. (166,44) -- cycle ;
\draw  [color={rgb, 255:red, 80; green, 227; blue, 194 }  ,draw opacity=1 ][fill={rgb, 255:red, 80; green, 227; blue, 194 }  ,fill opacity=1 ] (166,44) .. controls (166,40.13) and (169.13,37) .. (173,37) .. controls (176.87,37) and (180,40.13) .. (180,44) .. controls (180,47.87) and (176.87,51) .. (173,51) .. controls (169.13,51) and (166,47.87) .. (166,44) -- cycle ;
\draw  [color={rgb, 255:red, 80; green, 227; blue, 194 }  ,draw opacity=1 ][fill={rgb, 255:red, 80; green, 227; blue, 194 }  ,fill opacity=1 ] (166,44) .. controls (166,40.13) and (169.13,37) .. (173,37) .. controls (176.87,37) and (180,40.13) .. (180,44) .. controls (180,47.87) and (176.87,51) .. (173,51) .. controls (169.13,51) and (166,47.87) .. (166,44) -- cycle ;
\draw  [color={rgb, 255:red, 80; green, 227; blue, 194 }  ,draw opacity=1 ][fill={rgb, 255:red, 80; green, 227; blue, 194 }  ,fill opacity=1 ] (166,44) .. controls (166,40.13) and (169.13,37) .. (173,37) .. controls (176.87,37) and (180,40.13) .. (180,44) .. controls (180,47.87) and (176.87,51) .. (173,51) .. controls (169.13,51) and (166,47.87) .. (166,44) -- cycle ;
\draw  [color={rgb, 255:red, 80; green, 227; blue, 194 }  ,draw opacity=1 ][fill={rgb, 255:red, 80; green, 227; blue, 194 }  ,fill opacity=1 ] (166,44) .. controls (166,40.13) and (169.13,37) .. (173,37) .. controls (176.87,37) and (180,40.13) .. (180,44) .. controls (180,47.87) and (176.87,51) .. (173,51) .. controls (169.13,51) and (166,47.87) .. (166,44) -- cycle ;
\draw  [color={rgb, 255:red, 80; green, 227; blue, 194 }  ,draw opacity=1 ][fill={rgb, 255:red, 80; green, 227; blue, 194 }  ,fill opacity=1 ] (220,30) .. controls (220,26.13) and (223.13,23) .. (227,23) .. controls (230.87,23) and (234,26.13) .. (234,30) .. controls (234,33.87) and (230.87,37) .. (227,37) .. controls (223.13,37) and (220,33.87) .. (220,30) -- cycle ;
\draw  [color={rgb, 255:red, 80; green, 227; blue, 194 }  ,draw opacity=1 ][fill={rgb, 255:red, 80; green, 227; blue, 194 }  ,fill opacity=1 ] (190,55) .. controls (190,51.13) and (193.13,48) .. (197,48) .. controls (200.87,48) and (204,51.13) .. (204,55) .. controls (204,58.87) and (200.87,62) .. (197,62) .. controls (193.13,62) and (190,58.87) .. (190,55) -- cycle ;
\draw  [color={rgb, 255:red, 80; green, 227; blue, 194 }  ,draw opacity=1 ][fill={rgb, 255:red, 80; green, 227; blue, 194 }  ,fill opacity=1 ] (488,26) .. controls (488,22.13) and (491.13,19) .. (495,19) .. controls (498.87,19) and (502,22.13) .. (502,26) .. controls (502,29.87) and (498.87,33) .. (495,33) .. controls (491.13,33) and (488,29.87) .. (488,26) -- cycle ;
\draw  [color={rgb, 255:red, 80; green, 227; blue, 194 }  ,draw opacity=1 ][fill={rgb, 255:red, 80; green, 227; blue, 194 }  ,fill opacity=1 ] (409,26) .. controls (409,22.13) and (412.13,19) .. (416,19) .. controls (419.87,19) and (423,22.13) .. (423,26) .. controls (423,29.87) and (419.87,33) .. (416,33) .. controls (412.13,33) and (409,29.87) .. (409,26) -- cycle ;
\draw  [color={rgb, 255:red, 80; green, 227; blue, 194 }  ,draw opacity=1 ][fill={rgb, 255:red, 80; green, 227; blue, 194 }  ,fill opacity=1 ] (374,30) .. controls (374,26.13) and (377.13,23) .. (381,23) .. controls (384.87,23) and (388,26.13) .. (388,30) .. controls (388,33.87) and (384.87,37) .. (381,37) .. controls (377.13,37) and (374,33.87) .. (374,30) -- cycle ;
\draw  [color={rgb, 255:red, 80; green, 227; blue, 194 }  ,draw opacity=1 ][fill={rgb, 255:red, 80; green, 227; blue, 194 }  ,fill opacity=1 ] (452,38) .. controls (452,34.13) and (455.13,31) .. (459,31) .. controls (462.87,31) and (466,34.13) .. (466,38) .. controls (466,41.87) and (462.87,45) .. (459,45) .. controls (455.13,45) and (452,41.87) .. (452,38) -- cycle ;
\draw  [color={rgb, 255:red, 80; green, 227; blue, 194 }  ,draw opacity=1 ][fill={rgb, 255:red, 80; green, 227; blue, 194 }  ,fill opacity=1 ] (362,48) .. controls (362,44.13) and (365.13,41) .. (369,41) .. controls (372.87,41) and (376,44.13) .. (376,48) .. controls (376,51.87) and (372.87,55) .. (369,55) .. controls (365.13,55) and (362,51.87) .. (362,48) -- cycle ;
\draw  [color={rgb, 255:red, 80; green, 227; blue, 194 }  ,draw opacity=1 ][fill={rgb, 255:red, 80; green, 227; blue, 194 }  ,fill opacity=1 ] (328,56) .. controls (328,52.13) and (331.13,49) .. (335,49) .. controls (338.87,49) and (342,52.13) .. (342,56) .. controls (342,59.87) and (338.87,63) .. (335,63) .. controls (331.13,63) and (328,59.87) .. (328,56) -- cycle ;
\draw  [color={rgb, 255:red, 80; green, 227; blue, 194 }  ,draw opacity=1 ][fill={rgb, 255:red, 80; green, 227; blue, 194 }  ,fill opacity=1 ] (339,32) .. controls (339,28.13) and (342.13,25) .. (346,25) .. controls (349.87,25) and (353,28.13) .. (353,32) .. controls (353,35.87) and (349.87,39) .. (346,39) .. controls (342.13,39) and (339,35.87) .. (339,32) -- cycle ;
\draw  [color={rgb, 255:red, 80; green, 227; blue, 194 }  ,draw opacity=1 ][fill={rgb, 255:red, 80; green, 227; blue, 194 }  ,fill opacity=1 ] (299,27) .. controls (299,23.13) and (302.13,20) .. (306,20) .. controls (309.87,20) and (313,23.13) .. (313,27) .. controls (313,30.87) and (309.87,34) .. (306,34) .. controls (302.13,34) and (299,30.87) .. (299,27) -- cycle ;
\draw  [color={rgb, 255:red, 80; green, 227; blue, 194 }  ,draw opacity=1 ][fill={rgb, 255:red, 80; green, 227; blue, 194 }  ,fill opacity=1 ] (229,56) .. controls (229,52.13) and (232.13,49) .. (236,49) .. controls (239.87,49) and (243,52.13) .. (243,56) .. controls (243,59.87) and (239.87,63) .. (236,63) .. controls (232.13,63) and (229,59.87) .. (229,56) -- cycle ;
\draw  [color={rgb, 255:red, 0; green, 0; blue, 0 }  ,draw opacity=1 ][fill={rgb, 255:red, 213; green, 45; blue, 66 }  ,fill opacity=1 ] (245,26) .. controls (245,22.13) and (248.13,19) .. (252,19) .. controls (255.87,19) and (259,22.13) .. (259,26) .. controls (259,29.87) and (255.87,33) .. (252,33) .. controls (248.13,33) and (245,29.87) .. (245,26) -- cycle ;
\draw  [color={rgb, 255:red, 0; green, 0; blue, 0 }  ,draw opacity=1 ][fill={rgb, 255:red, 213; green, 45; blue, 66 }  ,fill opacity=1 ] (201,38) .. controls (201,34.13) and (204.13,31) .. (208,31) .. controls (211.87,31) and (215,34.13) .. (215,38) .. controls (215,41.87) and (211.87,45) .. (208,45) .. controls (204.13,45) and (201,41.87) .. (201,38) -- cycle ;
\draw  [color={rgb, 255:red, 0; green, 0; blue, 0 }  ,draw opacity=1 ][fill={rgb, 255:red, 213; green, 45; blue, 66 }  ,fill opacity=1 ] (246,46) .. controls (246,42.13) and (249.13,39) .. (253,39) .. controls (256.87,39) and (260,42.13) .. (260,46) .. controls (260,49.87) and (256.87,53) .. (253,53) .. controls (249.13,53) and (246,49.87) .. (246,46) -- cycle ;
\draw  [color={rgb, 255:red, 0; green, 0; blue, 0 }  ,draw opacity=1 ][fill={rgb, 255:red, 213; green, 45; blue, 66 }  ,fill opacity=1 ] (395,42) .. controls (395,38.13) and (398.13,35) .. (402,35) .. controls (405.87,35) and (409,38.13) .. (409,42) .. controls (409,45.87) and (405.87,49) .. (402,49) .. controls (398.13,49) and (395,45.87) .. (395,42) -- cycle ;
\draw  [color={rgb, 255:red, 0; green, 0; blue, 0 }  ,draw opacity=1 ][fill={rgb, 255:red, 213; green, 45; blue, 66 }  ,fill opacity=1 ] (429,42) .. controls (429,38.13) and (432.13,35) .. (436,35) .. controls (439.87,35) and (443,38.13) .. (443,42) .. controls (443,45.87) and (439.87,49) .. (436,49) .. controls (432.13,49) and (429,45.87) .. (429,42) -- cycle ;
\draw  [color={rgb, 255:red, 0; green, 0; blue, 0 }  ,draw opacity=1 ][fill={rgb, 255:red, 213; green, 45; blue, 66 }  ,fill opacity=1 ] (477,48) .. controls (477,44.13) and (480.13,41) .. (484,41) .. controls (487.87,41) and (491,44.13) .. (491,48) .. controls (491,51.87) and (487.87,55) .. (484,55) .. controls (480.13,55) and (477,51.87) .. (477,48) -- cycle ;
\draw  [color={rgb, 255:red, 0; green, 0; blue, 0 }  ,draw opacity=1 ][fill={rgb, 255:red, 213; green, 45; blue, 66 }  ,fill opacity=1 ] (414,57) .. controls (414,53.13) and (417.13,50) .. (421,50) .. controls (424.87,50) and (428,53.13) .. (428,57) .. controls (428,60.87) and (424.87,64) .. (421,64) .. controls (417.13,64) and (414,60.87) .. (414,57) -- cycle ;
\draw  [color={rgb, 255:red, 0; green, 0; blue, 0 }  ,draw opacity=1 ][fill={rgb, 255:red, 213; green, 45; blue, 66 }  ,fill opacity=1 ] (436,24) .. controls (436,20.13) and (439.13,17) .. (443,17) .. controls (446.87,17) and (450,20.13) .. (450,24) .. controls (450,27.87) and (446.87,31) .. (443,31) .. controls (439.13,31) and (436,27.87) .. (436,24) -- cycle ;
\draw  [color={rgb, 255:red, 80; green, 227; blue, 194 }  ,draw opacity=1 ][fill={rgb, 255:red, 80; green, 227; blue, 194 }  ,fill opacity=1 ] (138,155) .. controls (138,151.13) and (141.13,148) .. (145,148) .. controls (148.87,148) and (152,151.13) .. (152,155) .. controls (152,158.87) and (148.87,162) .. (145,162) .. controls (141.13,162) and (138,158.87) .. (138,155) -- cycle ;
\draw  [color={rgb, 255:red, 80; green, 227; blue, 194 }  ,draw opacity=1 ][fill={rgb, 255:red, 80; green, 227; blue, 194 }  ,fill opacity=1 ] (158,175) .. controls (158,171.13) and (161.13,168) .. (165,168) .. controls (168.87,168) and (172,171.13) .. (172,175) .. controls (172,178.87) and (168.87,182) .. (165,182) .. controls (161.13,182) and (158,178.87) .. (158,175) -- cycle ;
\draw  [color={rgb, 255:red, 80; green, 227; blue, 194 }  ,draw opacity=1 ][fill={rgb, 255:red, 80; green, 227; blue, 194 }  ,fill opacity=1 ] (246,157) .. controls (246,153.13) and (249.13,150) .. (253,150) .. controls (256.87,150) and (260,153.13) .. (260,157) .. controls (260,160.87) and (256.87,164) .. (253,164) .. controls (249.13,164) and (246,160.87) .. (246,157) -- cycle ;
\draw  [color={rgb, 255:red, 80; green, 227; blue, 194 }  ,draw opacity=1 ][fill={rgb, 255:red, 80; green, 227; blue, 194 }  ,fill opacity=1 ] (214,154) .. controls (214,150.13) and (217.13,147) .. (221,147) .. controls (224.87,147) and (228,150.13) .. (228,154) .. controls (228,157.87) and (224.87,161) .. (221,161) .. controls (217.13,161) and (214,157.87) .. (214,154) -- cycle ;
\draw  [color={rgb, 255:red, 80; green, 227; blue, 194 }  ,draw opacity=1 ][fill={rgb, 255:red, 80; green, 227; blue, 194 }  ,fill opacity=1 ] (202,187) .. controls (202,183.13) and (205.13,180) .. (209,180) .. controls (212.87,180) and (216,183.13) .. (216,187) .. controls (216,190.87) and (212.87,194) .. (209,194) .. controls (205.13,194) and (202,190.87) .. (202,187) -- cycle ;
\draw  [color={rgb, 255:red, 80; green, 227; blue, 194 }  ,draw opacity=1 ][fill={rgb, 255:red, 80; green, 227; blue, 194 }  ,fill opacity=1 ] (340,153) .. controls (340,149.13) and (343.13,146) .. (347,146) .. controls (350.87,146) and (354,149.13) .. (354,153) .. controls (354,156.87) and (350.87,160) .. (347,160) .. controls (343.13,160) and (340,156.87) .. (340,153) -- cycle ;
\draw  [color={rgb, 255:red, 80; green, 227; blue, 194 }  ,draw opacity=1 ][fill={rgb, 255:red, 80; green, 227; blue, 194 }  ,fill opacity=1 ] (414,153) .. controls (414,149.13) and (417.13,146) .. (421,146) .. controls (424.87,146) and (428,149.13) .. (428,153) .. controls (428,156.87) and (424.87,160) .. (421,160) .. controls (417.13,160) and (414,156.87) .. (414,153) -- cycle ;
\draw  [color={rgb, 255:red, 80; green, 227; blue, 194 }  ,draw opacity=1 ][fill={rgb, 255:red, 80; green, 227; blue, 194 }  ,fill opacity=1 ] (360,173) .. controls (360,169.13) and (363.13,166) .. (367,166) .. controls (370.87,166) and (374,169.13) .. (374,173) .. controls (374,176.87) and (370.87,180) .. (367,180) .. controls (363.13,180) and (360,176.87) .. (360,173) -- cycle ;
\draw  [color={rgb, 255:red, 80; green, 227; blue, 194 }  ,draw opacity=1 ][fill={rgb, 255:red, 80; green, 227; blue, 194 }  ,fill opacity=1 ] (424,183) .. controls (424,179.13) and (427.13,176) .. (431,176) .. controls (434.87,176) and (438,179.13) .. (438,183) .. controls (438,186.87) and (434.87,190) .. (431,190) .. controls (427.13,190) and (424,186.87) .. (424,183) -- cycle ;
\draw  [color={rgb, 255:red, 80; green, 227; blue, 194 }  ,draw opacity=1 ][fill={rgb, 255:red, 80; green, 227; blue, 194 }  ,fill opacity=1 ] (487,154) .. controls (487,150.13) and (490.13,147) .. (494,147) .. controls (497.87,147) and (501,150.13) .. (501,154) .. controls (501,157.87) and (497.87,161) .. (494,161) .. controls (490.13,161) and (487,157.87) .. (487,154) -- cycle ;
\draw  [color={rgb, 255:red, 80; green, 227; blue, 194 }  ,draw opacity=1 ][fill={rgb, 255:red, 80; green, 227; blue, 194 }  ,fill opacity=1 ] (449,151) .. controls (449,147.13) and (452.13,144) .. (456,144) .. controls (459.87,144) and (463,147.13) .. (463,151) .. controls (463,154.87) and (459.87,158) .. (456,158) .. controls (452.13,158) and (449,154.87) .. (449,151) -- cycle ;
\draw  [color={rgb, 255:red, 80; green, 227; blue, 194 }  ,draw opacity=1 ][fill={rgb, 255:red, 80; green, 227; blue, 194 }  ,fill opacity=1 ] (453,183) .. controls (453,179.13) and (456.13,176) .. (460,176) .. controls (463.87,176) and (467,179.13) .. (467,183) .. controls (467,186.87) and (463.87,190) .. (460,190) .. controls (456.13,190) and (453,186.87) .. (453,183) -- cycle ;
\draw  [color={rgb, 255:red, 80; green, 227; blue, 194 }  ,draw opacity=1 ][fill={rgb, 255:red, 80; green, 227; blue, 194 }  ,fill opacity=1 ] (308,186) .. controls (308,182.13) and (311.13,179) .. (315,179) .. controls (318.87,179) and (322,182.13) .. (322,186) .. controls (322,189.87) and (318.87,193) .. (315,193) .. controls (311.13,193) and (308,189.87) .. (308,186) -- cycle ;
\draw  [color={rgb, 255:red, 74; green, 74; blue, 74 }  ,draw opacity=1 ] (449,269) .. controls (449,252.43) and (462.43,239) .. (479,239) .. controls (495.57,239) and (509,252.43) .. (509,269) .. controls (509,285.57) and (495.57,299) .. (479,299) .. controls (462.43,299) and (449,285.57) .. (449,269) -- cycle ;
\draw  [color={rgb, 255:red, 74; green, 74; blue, 74 }  ,draw opacity=1 ] (372,271) .. controls (372,254.43) and (385.43,241) .. (402,241) .. controls (418.57,241) and (432,254.43) .. (432,271) .. controls (432,287.57) and (418.57,301) .. (402,301) .. controls (385.43,301) and (372,287.57) .. (372,271) -- cycle ;
\draw  [color={rgb, 255:red, 74; green, 74; blue, 74 }  ,draw opacity=1 ] (303,271) .. controls (303,254.43) and (316.43,241) .. (333,241) .. controls (349.57,241) and (363,254.43) .. (363,271) .. controls (363,287.57) and (349.57,301) .. (333,301) .. controls (316.43,301) and (303,287.57) .. (303,271) -- cycle ;
\draw  [color={rgb, 255:red, 74; green, 74; blue, 74 }  ,draw opacity=1 ] (168,272) .. controls (168,255.43) and (181.43,242) .. (198,242) .. controls (214.57,242) and (228,255.43) .. (228,272) .. controls (228,288.57) and (214.57,302) .. (198,302) .. controls (181.43,302) and (168,288.57) .. (168,272) -- cycle ;
\draw  [color={rgb, 255:red, 0; green, 0; blue, 0 }  ,draw opacity=1 ][fill={rgb, 255:red, 213; green, 45; blue, 66 }  ,fill opacity=1 ] (170,155) .. controls (170,151.13) and (173.13,148) .. (177,148) .. controls (180.87,148) and (184,151.13) .. (184,155) .. controls (184,158.87) and (180.87,162) .. (177,162) .. controls (173.13,162) and (170,158.87) .. (170,155) -- cycle ;
\draw  [color={rgb, 255:red, 0; green, 0; blue, 0 }  ,draw opacity=1 ][fill={rgb, 255:red, 213; green, 45; blue, 66 }  ,fill opacity=1 ] (250,183) .. controls (250,179.13) and (253.13,176) .. (257,176) .. controls (260.87,176) and (264,179.13) .. (264,183) .. controls (264,186.87) and (260.87,190) .. (257,190) .. controls (253.13,190) and (250,186.87) .. (250,183) -- cycle ;
\draw  [color={rgb, 255:red, 0; green, 0; blue, 0 }  ,draw opacity=1 ][fill={rgb, 255:red, 213; green, 45; blue, 66 }  ,fill opacity=1 ] (222,174) .. controls (222,170.13) and (225.13,167) .. (229,167) .. controls (232.87,167) and (236,170.13) .. (236,174) .. controls (236,177.87) and (232.87,181) .. (229,181) .. controls (225.13,181) and (222,177.87) .. (222,174) -- cycle ;
\draw  [color={rgb, 255:red, 0; green, 0; blue, 0 }  ,draw opacity=1 ][fill={rgb, 255:red, 213; green, 45; blue, 66 }  ,fill opacity=1 ] (139,185) .. controls (139,181.13) and (142.13,178) .. (146,178) .. controls (149.87,178) and (153,181.13) .. (153,185) .. controls (153,188.87) and (149.87,192) .. (146,192) .. controls (142.13,192) and (139,188.87) .. (139,185) -- cycle ;
\draw  [color={rgb, 255:red, 0; green, 0; blue, 0 }  ,draw opacity=1 ][fill={rgb, 255:red, 213; green, 45; blue, 66 }  ,fill opacity=1 ] (310,152) .. controls (310,148.13) and (313.13,145) .. (317,145) .. controls (320.87,145) and (324,148.13) .. (324,152) .. controls (324,155.87) and (320.87,159) .. (317,159) .. controls (313.13,159) and (310,155.87) .. (310,152) -- cycle ;
\draw  [color={rgb, 255:red, 80; green, 227; blue, 194 }  ,draw opacity=1 ][fill={rgb, 255:red, 80; green, 227; blue, 194 }  ,fill opacity=1 ] (176,188) .. controls (176,184.13) and (179.13,181) .. (183,181) .. controls (186.87,181) and (190,184.13) .. (190,188) .. controls (190,191.87) and (186.87,195) .. (183,195) .. controls (179.13,195) and (176,191.87) .. (176,188) -- cycle ;
\draw  [color={rgb, 255:red, 80; green, 227; blue, 194 }  ,draw opacity=1 ][fill={rgb, 255:red, 80; green, 227; blue, 194 }  ,fill opacity=1 ] (334,185) .. controls (334,181.13) and (337.13,178) .. (341,178) .. controls (344.87,178) and (348,181.13) .. (348,185) .. controls (348,188.87) and (344.87,192) .. (341,192) .. controls (337.13,192) and (334,188.87) .. (334,185) -- cycle ;
\draw  [color={rgb, 255:red, 0; green, 0; blue, 0 }  ,draw opacity=1 ][fill={rgb, 255:red, 213; green, 45; blue, 66 }  ,fill opacity=1 ] (325,167) .. controls (325,163.13) and (328.13,160) .. (332,160) .. controls (335.87,160) and (339,163.13) .. (339,167) .. controls (339,170.87) and (335.87,174) .. (332,174) .. controls (328.13,174) and (325,170.87) .. (325,167) -- cycle ;
\draw  [color={rgb, 255:red, 0; green, 0; blue, 0 }  ,draw opacity=1 ][fill={rgb, 255:red, 213; green, 45; blue, 66 }  ,fill opacity=1 ] (377,153) .. controls (377,149.13) and (380.13,146) .. (384,146) .. controls (387.87,146) and (391,149.13) .. (391,153) .. controls (391,156.87) and (387.87,160) .. (384,160) .. controls (380.13,160) and (377,156.87) .. (377,153) -- cycle ;
\draw  [color={rgb, 255:red, 0; green, 0; blue, 0 }  ,draw opacity=1 ][fill={rgb, 255:red, 213; green, 45; blue, 66 }  ,fill opacity=1 ] (490,184) .. controls (490,180.13) and (493.13,177) .. (497,177) .. controls (500.87,177) and (504,180.13) .. (504,184) .. controls (504,187.87) and (500.87,191) .. (497,191) .. controls (493.13,191) and (490,187.87) .. (490,184) -- cycle ;
\draw  [color={rgb, 255:red, 0; green, 0; blue, 0 }  ,draw opacity=1 ][fill={rgb, 255:red, 213; green, 45; blue, 66 }  ,fill opacity=1 ] (465,169) .. controls (465,165.13) and (468.13,162) .. (472,162) .. controls (475.87,162) and (479,165.13) .. (479,169) .. controls (479,172.87) and (475.87,176) .. (472,176) .. controls (468.13,176) and (465,172.87) .. (465,169) -- cycle ;
\draw  [color={rgb, 255:red, 0; green, 0; blue, 0 }  ,draw opacity=1 ][fill={rgb, 255:red, 213; green, 45; blue, 66 }  ,fill opacity=1 ] (402,187) .. controls (402,183.13) and (405.13,180) .. (409,180) .. controls (412.87,180) and (416,183.13) .. (416,187) .. controls (416,190.87) and (412.87,194) .. (409,194) .. controls (405.13,194) and (402,190.87) .. (402,187) -- cycle ;
\draw  [color={rgb, 255:red, 80; green, 227; blue, 194 }  ,draw opacity=1 ][fill={rgb, 255:red, 80; green, 227; blue, 194 }  ,fill opacity=1 ] (378,185) .. controls (378,181.13) and (381.13,178) .. (385,178) .. controls (388.87,178) and (392,181.13) .. (392,185) .. controls (392,188.87) and (388.87,192) .. (385,192) .. controls (381.13,192) and (378,188.87) .. (378,185) -- cycle ;
\draw  [color={rgb, 255:red, 0; green, 0; blue, 0 }  ,draw opacity=1 ][fill={rgb, 255:red, 213; green, 45; blue, 66 }  ,fill opacity=1 ] (322,253) .. controls (322,249.13) and (325.13,246) .. (329,246) .. controls (332.87,246) and (336,249.13) .. (336,253) .. controls (336,256.87) and (332.87,260) .. (329,260) .. controls (325.13,260) and (322,256.87) .. (322,253) -- cycle ;
\draw  [color={rgb, 255:red, 0; green, 0; blue, 0 }  ,draw opacity=1 ][fill={rgb, 255:red, 213; green, 45; blue, 66 }  ,fill opacity=1 ] (174,271) .. controls (174,267.13) and (177.13,264) .. (181,264) .. controls (184.87,264) and (188,267.13) .. (188,271) .. controls (188,274.87) and (184.87,278) .. (181,278) .. controls (177.13,278) and (174,274.87) .. (174,271) -- cycle ;
\draw  [color={rgb, 255:red, 0; green, 0; blue, 0 }  ,draw opacity=1 ][fill={rgb, 255:red, 213; green, 45; blue, 66 }  ,fill opacity=1 ] (397,289) .. controls (397,285.13) and (400.13,282) .. (404,282) .. controls (407.87,282) and (411,285.13) .. (411,289) .. controls (411,292.87) and (407.87,296) .. (404,296) .. controls (400.13,296) and (397,292.87) .. (397,289) -- cycle ;
\draw  [color={rgb, 255:red, 0; green, 0; blue, 0 }  ,draw opacity=1 ][fill={rgb, 255:red, 213; green, 45; blue, 66 }  ,fill opacity=1 ] (378,278) .. controls (378,274.13) and (381.13,271) .. (385,271) .. controls (388.87,271) and (392,274.13) .. (392,278) .. controls (392,281.87) and (388.87,285) .. (385,285) .. controls (381.13,285) and (378,281.87) .. (378,278) -- cycle ;
\draw  [color={rgb, 255:red, 0; green, 0; blue, 0 }  ,draw opacity=1 ][fill={rgb, 255:red, 213; green, 45; blue, 66 }  ,fill opacity=1 ] (327,289) .. controls (327,285.13) and (330.13,282) .. (334,282) .. controls (337.87,282) and (341,285.13) .. (341,289) .. controls (341,292.87) and (337.87,296) .. (334,296) .. controls (330.13,296) and (327,292.87) .. (327,289) -- cycle ;
\draw  [color={rgb, 255:red, 0; green, 0; blue, 0 }  ,draw opacity=1 ][fill={rgb, 255:red, 213; green, 45; blue, 66 }  ,fill opacity=1 ] (455,274) .. controls (455,270.13) and (458.13,267) .. (462,267) .. controls (465.87,267) and (469,270.13) .. (469,274) .. controls (469,277.87) and (465.87,281) .. (462,281) .. controls (458.13,281) and (455,277.87) .. (455,274) -- cycle ;
\draw  [color={rgb, 255:red, 0; green, 0; blue, 0 }  ,draw opacity=1 ][fill={rgb, 255:red, 213; green, 45; blue, 66 }  ,fill opacity=1 ] (488,270) .. controls (488,266.13) and (491.13,263) .. (495,263) .. controls (498.87,263) and (502,266.13) .. (502,270) .. controls (502,273.87) and (498.87,277) .. (495,277) .. controls (491.13,277) and (488,273.87) .. (488,270) -- cycle ;
\draw  [color={rgb, 255:red, 0; green, 0; blue, 0 }  ,draw opacity=1 ][fill={rgb, 255:red, 213; green, 45; blue, 66 }  ,fill opacity=1 ] (182,287) .. controls (182,283.13) and (185.13,280) .. (189,280) .. controls (192.87,280) and (196,283.13) .. (196,287) .. controls (196,290.87) and (192.87,294) .. (189,294) .. controls (185.13,294) and (182,290.87) .. (182,287) -- cycle ;
\draw  [color={rgb, 255:red, 0; green, 0; blue, 0 }  ,draw opacity=1 ][fill={rgb, 255:red, 213; green, 45; blue, 66 }  ,fill opacity=1 ] (398,254) .. controls (398,250.13) and (401.13,247) .. (405,247) .. controls (408.87,247) and (412,250.13) .. (412,254) .. controls (412,257.87) and (408.87,261) .. (405,261) .. controls (401.13,261) and (398,257.87) .. (398,254) -- cycle ;
\draw  [color={rgb, 255:red, 80; green, 227; blue, 194 }  ,draw opacity=1 ][fill={rgb, 255:red, 80; green, 227; blue, 194 }  ,fill opacity=1 ] (310,280) .. controls (310,276.13) and (313.13,273) .. (317,273) .. controls (320.87,273) and (324,276.13) .. (324,280) .. controls (324,283.87) and (320.87,287) .. (317,287) .. controls (313.13,287) and (310,283.87) .. (310,280) -- cycle ;
\draw  [color={rgb, 255:red, 80; green, 227; blue, 194 }  ,draw opacity=1 ][fill={rgb, 255:red, 80; green, 227; blue, 194 }  ,fill opacity=1 ] (342,263) .. controls (342,259.13) and (345.13,256) .. (349,256) .. controls (352.87,256) and (356,259.13) .. (356,263) .. controls (356,266.87) and (352.87,270) .. (349,270) .. controls (345.13,270) and (342,266.87) .. (342,263) -- cycle ;
\draw  [color={rgb, 255:red, 80; green, 227; blue, 194 }  ,draw opacity=1 ][fill={rgb, 255:red, 80; green, 227; blue, 194 }  ,fill opacity=1 ] (188,256) .. controls (188,252.13) and (191.13,249) .. (195,249) .. controls (198.87,249) and (202,252.13) .. (202,256) .. controls (202,259.87) and (198.87,263) .. (195,263) .. controls (191.13,263) and (188,259.87) .. (188,256) -- cycle ;
\draw  [color={rgb, 255:red, 80; green, 227; blue, 194 }  ,draw opacity=1 ][fill={rgb, 255:red, 80; green, 227; blue, 194 }  ,fill opacity=1 ] (207,265) .. controls (207,261.13) and (210.13,258) .. (214,258) .. controls (217.87,258) and (221,261.13) .. (221,265) .. controls (221,268.87) and (217.87,272) .. (214,272) .. controls (210.13,272) and (207,268.87) .. (207,265) -- cycle ;
\draw  [color={rgb, 255:red, 80; green, 227; blue, 194 }  ,draw opacity=1 ][fill={rgb, 255:red, 80; green, 227; blue, 194 }  ,fill opacity=1 ] (206,284) .. controls (206,280.13) and (209.13,277) .. (213,277) .. controls (216.87,277) and (220,280.13) .. (220,284) .. controls (220,287.87) and (216.87,291) .. (213,291) .. controls (209.13,291) and (206,287.87) .. (206,284) -- cycle ;
\draw  [color={rgb, 255:red, 80; green, 227; blue, 194 }  ,draw opacity=1 ][fill={rgb, 255:red, 80; green, 227; blue, 194 }  ,fill opacity=1 ] (380,259) .. controls (380,255.13) and (383.13,252) .. (387,252) .. controls (390.87,252) and (394,255.13) .. (394,259) .. controls (394,262.87) and (390.87,266) .. (387,266) .. controls (383.13,266) and (380,262.87) .. (380,259) -- cycle ;
\draw  [color={rgb, 255:red, 80; green, 227; blue, 194 }  ,draw opacity=1 ][fill={rgb, 255:red, 80; green, 227; blue, 194 }  ,fill opacity=1 ] (413,277) .. controls (413,273.13) and (416.13,270) .. (420,270) .. controls (423.87,270) and (427,273.13) .. (427,277) .. controls (427,280.87) and (423.87,284) .. (420,284) .. controls (416.13,284) and (413,280.87) .. (413,277) -- cycle ;
\draw  [color={rgb, 255:red, 80; green, 227; blue, 194 }  ,draw opacity=1 ][fill={rgb, 255:red, 80; green, 227; blue, 194 }  ,fill opacity=1 ] (474,286) .. controls (474,282.13) and (477.13,279) .. (481,279) .. controls (484.87,279) and (488,282.13) .. (488,286) .. controls (488,289.87) and (484.87,293) .. (481,293) .. controls (477.13,293) and (474,289.87) .. (474,286) -- cycle ;
\draw  [color={rgb, 255:red, 80; green, 227; blue, 194 }  ,draw opacity=1 ][fill={rgb, 255:red, 80; green, 227; blue, 194 }  ,fill opacity=1 ] (480,253) .. controls (480,249.13) and (483.13,246) .. (487,246) .. controls (490.87,246) and (494,249.13) .. (494,253) .. controls (494,256.87) and (490.87,260) .. (487,260) .. controls (483.13,260) and (480,256.87) .. (480,253) -- cycle ;
\draw  [color={rgb, 255:red, 80; green, 227; blue, 194 }  ,draw opacity=1 ][fill={rgb, 255:red, 80; green, 227; blue, 194 }  ,fill opacity=1 ] (460,255) .. controls (460,251.13) and (463.13,248) .. (467,248) .. controls (470.87,248) and (474,251.13) .. (474,255) .. controls (474,258.87) and (470.87,262) .. (467,262) .. controls (463.13,262) and (460,258.87) .. (460,255) -- cycle ;
\draw   (177.81,337.44) -- (221,337.44) -- (221,373.18) -- (177.81,373.18) -- cycle ;
\draw   (177,401.78) -- (220.19,401.78) -- (220.19,440.22) -- (177,440.22) -- cycle ;
\draw   (177.81,463.45) -- (221,463.45) -- (221,501.88) -- (177.81,501.88) -- cycle ;
\draw    (198.19,401.19) -- (198.19,374.59) ;
\draw [shift={(198.19,372.59)}, rotate = 450] [color={rgb, 255:red, 0; green, 0; blue, 0 }  ][line width=0.75]    (10.93,-3.29) .. controls (6.95,-1.4) and (3.31,-0.3) .. (0,0) .. controls (3.31,0.3) and (6.95,1.4) .. (10.93,3.29)   ;
\draw    (199.81,462.86) -- (199.81,442.51) ;
\draw [shift={(199.81,440.51)}, rotate = 450] [color={rgb, 255:red, 0; green, 0; blue, 0 }  ][line width=0.75]    (10.93,-3.29) .. controls (6.95,-1.4) and (3.31,-0.3) .. (0,0) .. controls (3.31,0.3) and (6.95,1.4) .. (10.93,3.29)   ;
\draw    (199.81,523.33) -- (199.81,504.78) ;
\draw [shift={(199.81,502.78)}, rotate = 450] [color={rgb, 255:red, 0; green, 0; blue, 0 }  ][line width=0.75]    (10.93,-3.29) .. controls (6.95,-1.4) and (3.31,-0.3) .. (0,0) .. controls (3.31,0.3) and (6.95,1.4) .. (10.93,3.29)   ;
\draw    (199.81,337.73) -- (199.81,320.67) ;
\draw [shift={(199.81,318.67)}, rotate = 450] [color={rgb, 255:red, 0; green, 0; blue, 0 }  ][line width=0.75]    (10.93,-3.29) .. controls (6.95,-1.4) and (3.31,-0.3) .. (0,0) .. controls (3.31,0.3) and (6.95,1.4) .. (10.93,3.29)   ;
\draw   (320.57,337.26) -- (351,337.26) -- (351,366.12) -- (320.57,366.12) -- cycle ;
\draw   (320,389.21) -- (350.43,389.21) -- (350.43,420.23) -- (320,420.23) -- cycle ;
\draw   (320.57,438.99) -- (351,438.99) -- (351,470.02) -- (320.57,470.02) -- cycle ;
\draw    (334.93,388.73) -- (334.93,367.64) ;
\draw [shift={(334.93,365.64)}, rotate = 450] [color={rgb, 255:red, 0; green, 0; blue, 0 }  ][line width=0.75]    (10.93,-3.29) .. controls (6.95,-1.4) and (3.31,-0.3) .. (0,0) .. controls (3.31,0.3) and (6.95,1.4) .. (10.93,3.29)   ;
\draw    (336.07,438.51) -- (336.07,422.47) ;
\draw [shift={(336.07,420.47)}, rotate = 450] [color={rgb, 255:red, 0; green, 0; blue, 0 }  ][line width=0.75]    (10.93,-3.29) .. controls (6.95,-1.4) and (3.31,-0.3) .. (0,0) .. controls (3.31,0.3) and (6.95,1.4) .. (10.93,3.29)   ;
\draw    (336.07,487.33) -- (336.07,472.74) ;
\draw [shift={(336.07,470.74)}, rotate = 450] [color={rgb, 255:red, 0; green, 0; blue, 0 }  ][line width=0.75]    (10.93,-3.29) .. controls (6.95,-1.4) and (3.31,-0.3) .. (0,0) .. controls (3.31,0.3) and (6.95,1.4) .. (10.93,3.29)   ;
\draw    (335.5,337.26) -- (335,322.67) ;
\draw [shift={(334.93,320.67)}, rotate = 448.02] [color={rgb, 255:red, 0; green, 0; blue, 0 }  ][line width=0.75]    (10.93,-3.29) .. controls (6.95,-1.4) and (3.31,-0.3) .. (0,0) .. controls (3.31,0.3) and (6.95,1.4) .. (10.93,3.29)   ;
\draw   (388.57,338.26) -- (419,338.26) -- (419,367.12) -- (388.57,367.12) -- cycle ;
\draw   (388,390.21) -- (418.43,390.21) -- (418.43,421.23) -- (388,421.23) -- cycle ;
\draw   (388.57,439.99) -- (419,439.99) -- (419,471.02) -- (388.57,471.02) -- cycle ;
\draw    (402.93,389.73) -- (402.93,368.64) ;
\draw [shift={(402.93,366.64)}, rotate = 450] [color={rgb, 255:red, 0; green, 0; blue, 0 }  ][line width=0.75]    (10.93,-3.29) .. controls (6.95,-1.4) and (3.31,-0.3) .. (0,0) .. controls (3.31,0.3) and (6.95,1.4) .. (10.93,3.29)   ;
\draw    (404.07,439.51) -- (404.07,423.47) ;
\draw [shift={(404.07,421.47)}, rotate = 450] [color={rgb, 255:red, 0; green, 0; blue, 0 }  ][line width=0.75]    (10.93,-3.29) .. controls (6.95,-1.4) and (3.31,-0.3) .. (0,0) .. controls (3.31,0.3) and (6.95,1.4) .. (10.93,3.29)   ;
\draw    (404.07,488.33) -- (404.07,473.74) ;
\draw [shift={(404.07,471.74)}, rotate = 450] [color={rgb, 255:red, 0; green, 0; blue, 0 }  ][line width=0.75]    (10.93,-3.29) .. controls (6.95,-1.4) and (3.31,-0.3) .. (0,0) .. controls (3.31,0.3) and (6.95,1.4) .. (10.93,3.29)   ;
\draw    (403.5,338.26) -- (403,323.67) ;
\draw [shift={(402.93,321.67)}, rotate = 448.02] [color={rgb, 255:red, 0; green, 0; blue, 0 }  ][line width=0.75]    (10.93,-3.29) .. controls (6.95,-1.4) and (3.31,-0.3) .. (0,0) .. controls (3.31,0.3) and (6.95,1.4) .. (10.93,3.29)   ;
\draw   (459.57,337.26) -- (490,337.26) -- (490,366.12) -- (459.57,366.12) -- cycle ;
\draw   (459,389.21) -- (489.43,389.21) -- (489.43,420.23) -- (459,420.23) -- cycle ;
\draw   (459.57,438.99) -- (490,438.99) -- (490,470.02) -- (459.57,470.02) -- cycle ;
\draw    (473.93,388.73) -- (473.93,367.64) ;
\draw [shift={(473.93,365.64)}, rotate = 450] [color={rgb, 255:red, 0; green, 0; blue, 0 }  ][line width=0.75]    (10.93,-3.29) .. controls (6.95,-1.4) and (3.31,-0.3) .. (0,0) .. controls (3.31,0.3) and (6.95,1.4) .. (10.93,3.29)   ;
\draw    (475.07,438.51) -- (475.07,422.47) ;
\draw [shift={(475.07,420.47)}, rotate = 450] [color={rgb, 255:red, 0; green, 0; blue, 0 }  ][line width=0.75]    (10.93,-3.29) .. controls (6.95,-1.4) and (3.31,-0.3) .. (0,0) .. controls (3.31,0.3) and (6.95,1.4) .. (10.93,3.29)   ;
\draw    (475.07,487.33) -- (475.07,472.74) ;
\draw [shift={(475.07,470.74)}, rotate = 450] [color={rgb, 255:red, 0; green, 0; blue, 0 }  ][line width=0.75]    (10.93,-3.29) .. controls (6.95,-1.4) and (3.31,-0.3) .. (0,0) .. controls (3.31,0.3) and (6.95,1.4) .. (10.93,3.29)   ;
\draw    (474.5,337.26) -- (474,322.67) ;
\draw [shift={(473.93,320.67)}, rotate = 448.02] [color={rgb, 255:red, 0; green, 0; blue, 0 }  ][line width=0.75]    (10.93,-3.29) .. controls (6.95,-1.4) and (3.31,-0.3) .. (0,0) .. controls (3.31,0.3) and (6.95,1.4) .. (10.93,3.29)   ;
\draw  [color={rgb, 255:red, 80; green, 227; blue, 194 }  ,draw opacity=1 ][fill={rgb, 255:red, 80; green, 227; blue, 194 }  ,fill opacity=1 ] (164,593) .. controls (164,589.13) and (167.13,586) .. (171,586) .. controls (174.87,586) and (178,589.13) .. (178,593) .. controls (178,596.87) and (174.87,600) .. (171,600) .. controls (167.13,600) and (164,596.87) .. (164,593) -- cycle ;
\draw  [color={rgb, 255:red, 0; green, 0; blue, 0 }  ,draw opacity=1 ][fill={rgb, 255:red, 213; green, 45; blue, 66 }  ,fill opacity=1 ] (322,593) .. controls (322,589.13) and (325.13,586) .. (329,586) .. controls (332.87,586) and (336,589.13) .. (336,593) .. controls (336,596.87) and (332.87,600) .. (329,600) .. controls (325.13,600) and (322,596.87) .. (322,593) -- cycle ;
\draw    (322,68.33) .. controls (334.74,98.71) and (389.74,113.73) .. (411.7,137.85) ;
\draw [shift={(413,139.33)}, rotate = 229.97] [color={rgb, 255:red, 0; green, 0; blue, 0 }  ][line width=0.75]    (10.93,-3.29) .. controls (6.95,-1.4) and (3.31,-0.3) .. (0,0) .. controls (3.31,0.3) and (6.95,1.4) .. (10.93,3.29)   ;
\draw   (604.72,38.93) -- (602.91,184.92) -- (543.51,184.18) .. controls (544.64,92.94) and (523,110.92) .. (537.76,38.1) -- cycle ;
\draw    (300,68.33) .. controls (283.34,102.63) and (229.22,123.49) .. (201.65,140.31) ;
\draw [shift={(200,141.33)}, rotate = 327.8] [color={rgb, 255:red, 0; green, 0; blue, 0 }  ][line width=0.75]    (10.93,-3.29) .. controls (6.95,-1.4) and (3.31,-0.3) .. (0,0) .. controls (3.31,0.3) and (6.95,1.4) .. (10.93,3.29)   ;
\draw   (36.79,185.06) -- (38.29,39.07) -- (97.69,39.68) .. controls (96.75,130.93) and (118.36,112.9) .. (103.75,185.75) -- cycle ;
\draw    (103,101.67) .. controls (142.8,71.82) and (193.49,88.5) .. (258.03,108.37) ;
\draw [shift={(259,108.67)}, rotate = 197.1] [color={rgb, 255:red, 0; green, 0; blue, 0 }  ][line width=0.75]    (10.93,-3.29) .. controls (6.95,-1.4) and (3.31,-0.3) .. (0,0) .. controls (3.31,0.3) and (6.95,1.4) .. (10.93,3.29)   ;
\draw    (534,90.33) .. controls (460.74,74.49) and (427.66,107.66) .. (371.7,110.26) ;
\draw [shift={(370,110.33)}, rotate = 357.99] [color={rgb, 255:red, 0; green, 0; blue, 0 }  ][line width=0.75]    (10.93,-3.29) .. controls (6.95,-1.4) and (3.31,-0.3) .. (0,0) .. controls (3.31,0.3) and (6.95,1.4) .. (10.93,3.29)   ;

\draw (311,501.33) node [anchor=north west][inner sep=0.75pt]   [align=left] {Shard 1};
\draw (448,501.33) node [anchor=north west][inner sep=0.75pt]   [align=left] {Shard n};
\draw (165,532.67) node [anchor=north west][inner sep=0.75pt]   [align=left] {Beacon Chain};
\draw (392,499.67) node [anchor=north west][inner sep=0.75pt]   [align=left] {. . . .};
\draw (190,585.67) node [anchor=north west][inner sep=0.75pt]   [align=left] {Honest Node};
\draw (349,586.67) node [anchor=north west][inner sep=0.75pt]   [align=left] {Byzantine Node};
\draw (602.32,43.9) node [anchor=north west][inner sep=0.75pt]  [rotate=-90.71] [align=left] {Users (nodes) stake\\to be validators to \\the Beacon chain};
\draw (39.18,180.09) node [anchor=north west][inner sep=0.75pt]  [rotate=-270.59] [align=left] {Users (nodes) stake\\to be validators to \\the shards chains};
\end{tikzpicture}
\caption{A scheme of a sharding-based PoS and pBFT Blockchain protocol}
\label{fig:sharding_scheme}
\end{figure*}

\begin{lemma} \label{lemma 1}
The probability of a shard's committee members committing a faulty block ($P$) can be described as follows: 
\begin{equation} 
    P(X \geq \frac{2n}{3} ) = \sum_{k = \frac{2n}{3} }^{n} \frac{ \binom{M}{k}
    \binom{H}{n-k}} {\binom{V}{n} }
\end{equation} 
\end{lemma}

\begin{lemma} \label{lemma 2}
The probability of at least $\frac{2}{3}$ of all shards committees committing a faulty block ($P^{'}$) can be computed as follows: 
\begin{equation} 
\sum_{i = \frac{2 \zeta}{3} }^{\zeta} \bigg( P(X \geq \frac{2n}{3}) \bigg)^{i} = \sum_{i = \frac{2 \zeta}{3} }^{\zeta} \sum_{\alpha = \frac{2n}{3} }^{n} \Bigg( \frac{ \binom{M}{\alpha} \binom{H}{n- \alpha}} {\binom{V}{n}} \Bigg)^{i}
\end{equation}
\end{lemma} 

\begin{proof}
The minimum of all shards' committees to commit a faulty block is $\frac{2 \zeta}{3}$, where $\zeta$ is the number of shards. The probability of exactly $\frac{2 \zeta}{3}$ all shards' committees confirm/agree to add a faulty block can be summarized and expressed as follows:
\begin{equation}
P_{\frac{2 \zeta}{3}} = \bigg( P(X \geq \frac{2n}{3}) \bigg)^{\frac{2 \zeta}{3}}   
\end{equation}

And the probability to commit a faulty block by exactly $\frac{2 \zeta}{3} + 1$ of all shards' committees can be described as follows.
\begin{equation}
P_{\frac{2 \zeta}{3} + 1} = \bigg( P(X \geq \frac{2n}{3}) \bigg)^{\frac{2 \zeta}{3} + 1}   
\end{equation}

Similarly, the probability of exactly $\zeta$ of all shards' committees (the entire number of shards in this case) agree to add a faulty block can be expressed as follows:

\begin{equation}
P_{\zeta} = \bigg( P(X \geq \frac{2n}{3}) \bigg)^{\zeta}   
\end{equation}

A faulty block can be committed if $\frac{2 \zeta}{3}$ or $\frac{2 \zeta}{3} + 1$ or $\frac{2 \zeta}{3} + 2$ or, $\cdots$, or $\zeta$ of all shards' committees agree to add this block. This can be mathematically computed by the sum over all these probabilities and can be expressed and summarized as follows:
\begin{equation}
P^{'} = P_{\frac{2 \zeta}{3}} + P_{\frac{2 \zeta}{3} + 1} + \cdots + P_{\zeta} 
\end{equation}
\end{proof} 

\begin{lemma} \label{lemma 3}
The probability of the beacon's committee members committing a faulty block ($P^{''}$) can be expressed as follows: 
\begin{equation} 
    P(X' \geq \frac{2n'}{3} ) = \sum_{j = \frac{2 n'}{3}  }^{n'} \frac{ \binom{M'}{j}
    \binom{H'}{n' - j}} {\binom{V'}{n'}}
\end{equation}
\end{lemma}

Proofs of Lemma \ref{lemma 1} and \ref{lemma 2} are a direct results from the cumulative hypergeometric distribution \cite{hafid2020novel, hafid2021tractable}.

\begin{theorem}[Committing a Faulty Block] \label{theorem 1}
The probability of committing a faulty block for a given shard can be expressed as follows: 
\begin{equation} 
    \mathcal{P} = \sum_{k = \frac{2n}{3} }^{n} \sum_{i = \frac{2 \zeta}{3} }^{\zeta} \sum_{\alpha = \frac{2n}{3} }^{n} \sum_{j = \frac{2 n'}{3}  }^{n'} \frac{ \binom{M}{k} \binom{H}{n-k} \binom{M}{\alpha}^{i} \binom{H}{n - \alpha}^{i} \binom{M'}{j} \binom{H'}{n' - j}  } {\binom{V}{n} \binom{H}{n- \alpha}^{i} \binom{V'}{n'}}
\end{equation}
\end{theorem}

\begin{proof}
To commit a faulty block, it must be confirmed/verified by at least $\frac{2}{3}$ of a shard's committee members and by at least $\frac{2}{3}$ of beacon's committee members, and finally by at least $\frac{2}{3}$ of all shards' committees. This can be identified mathematically by the product over the three probabilities (the calculated probabilities in Lemmas \ref{lemma 1}, \ref{lemma 2}, and \ref{lemma 3}). 
\end{proof}

\section{Results \& Evaluation} \label{sec: Results and Evaluation}
In Figures \ref{fig: Probability of a shard to commit a faulty block}, \ref{fig: Probability of all shards committing a faulty block}, and \ref{fig: Probability of the beacon chain to commit a faulty block}, we assume a network with N = 2000 users/nodes, V = 200, V'= 400, r = 0.333, R = 0.1, R = 0.2, and R = 0.3. The remaining parameters will be shown in Figures.

\begin{figure}[ht] 
\centering
\includegraphics[width= 0.5 \textwidth]{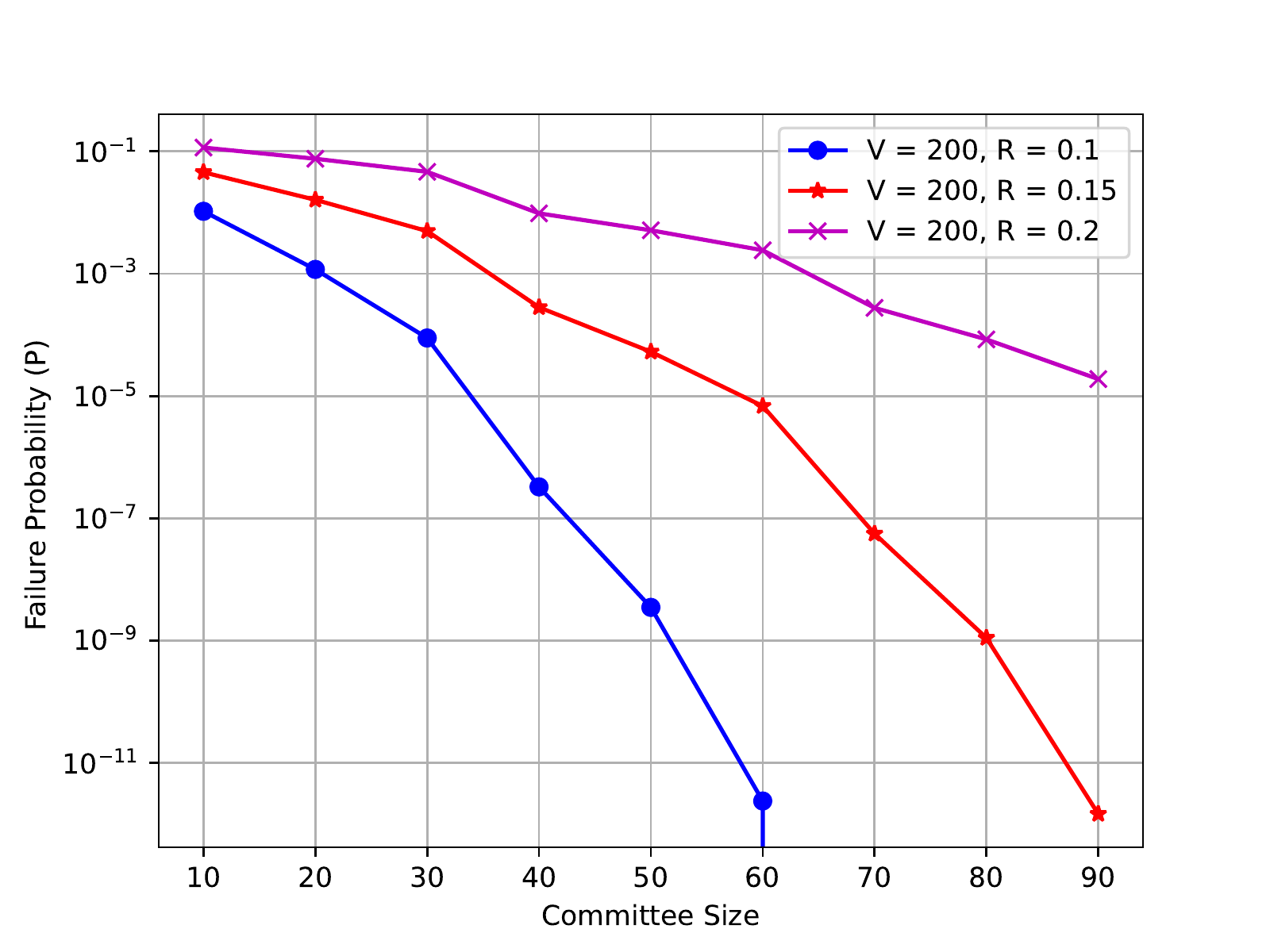}
\caption{Log-scale plot of the probability of a shard to commit a faulty block ($P$) versus the size of the committee (n)}
\label{fig: Probability of a shard to commit a faulty block}
\end{figure}

Figure \ref{fig: Probability of a shard to commit a faulty block} shows the probability of a shard to commit a faulty block versus the size of the committee in a network of 2000 nodes. We observe that the probability $P$ decreases when the size of the committee increases. More specifically, we observe that the probability corresponding to $R=0.1$ (i.e. 10\% of malicious nodes in each shards) decreases rapidly compared to those of $R=0.15$ and $R=0.2$; this can be explained by the small percentage of malicious nodes. In other words, as the percentage of malicious nodes gets smaller the probability decreases and vice versa.
 
\begin{figure}[ht] 
\centering
\includegraphics[width= 0.5 \textwidth]{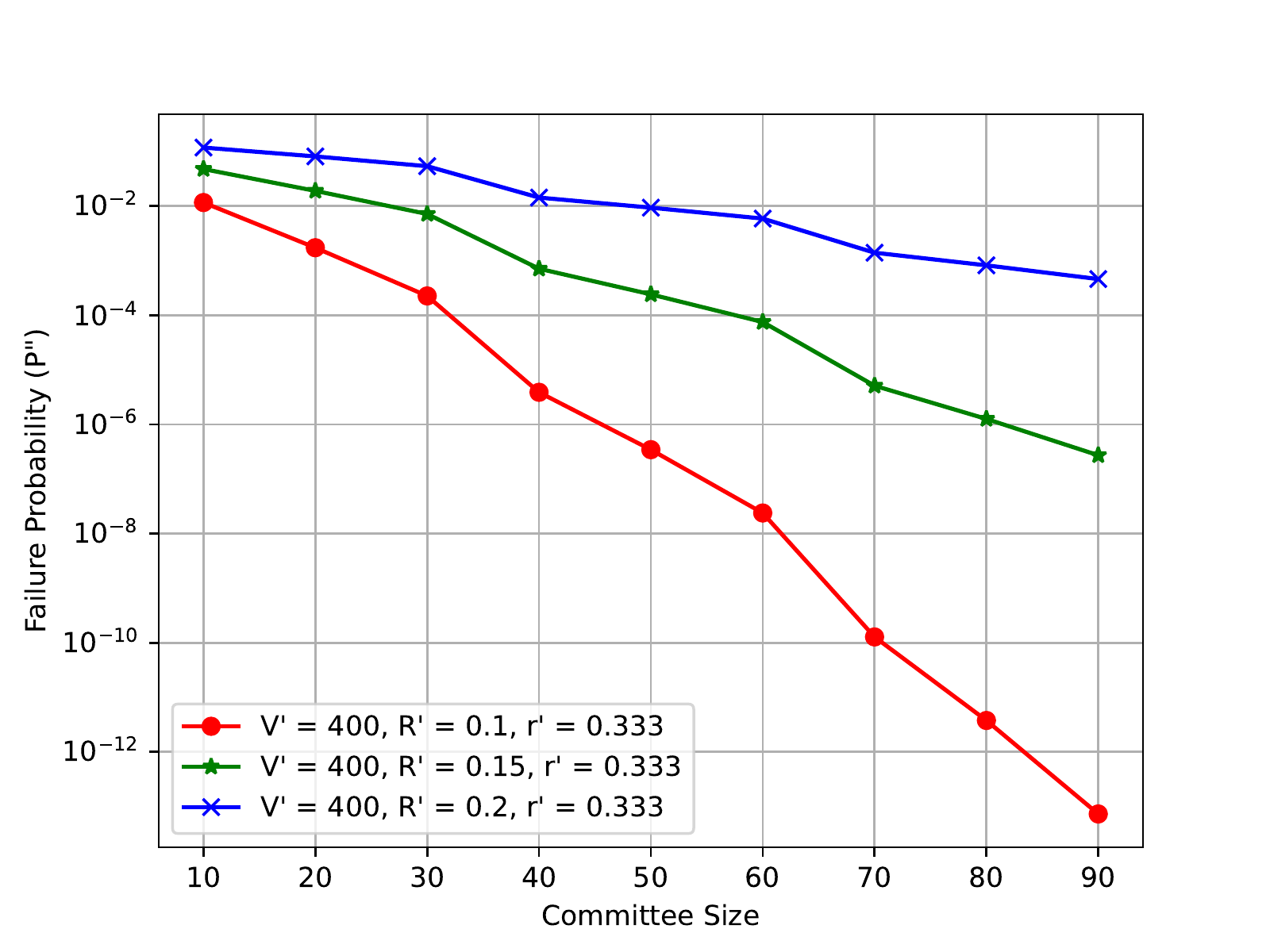}
\caption{Log-scale plot of the probability of all shards committing a faulty block ($P^{'}$) versus the size of the committee (n)}
\label{fig: Probability of the beacon chain to commit a faulty block}
\end{figure}

Figure \ref{fig: Probability of all shards committing a faulty block} shows the probability of all shards committing a faulty block versus the size of the committee. We observe that the probability $P^{'}$ decreases when the size of the committee increases.

\begin{figure}[ht] 
\centering
\includegraphics[width= 0.5 \textwidth]{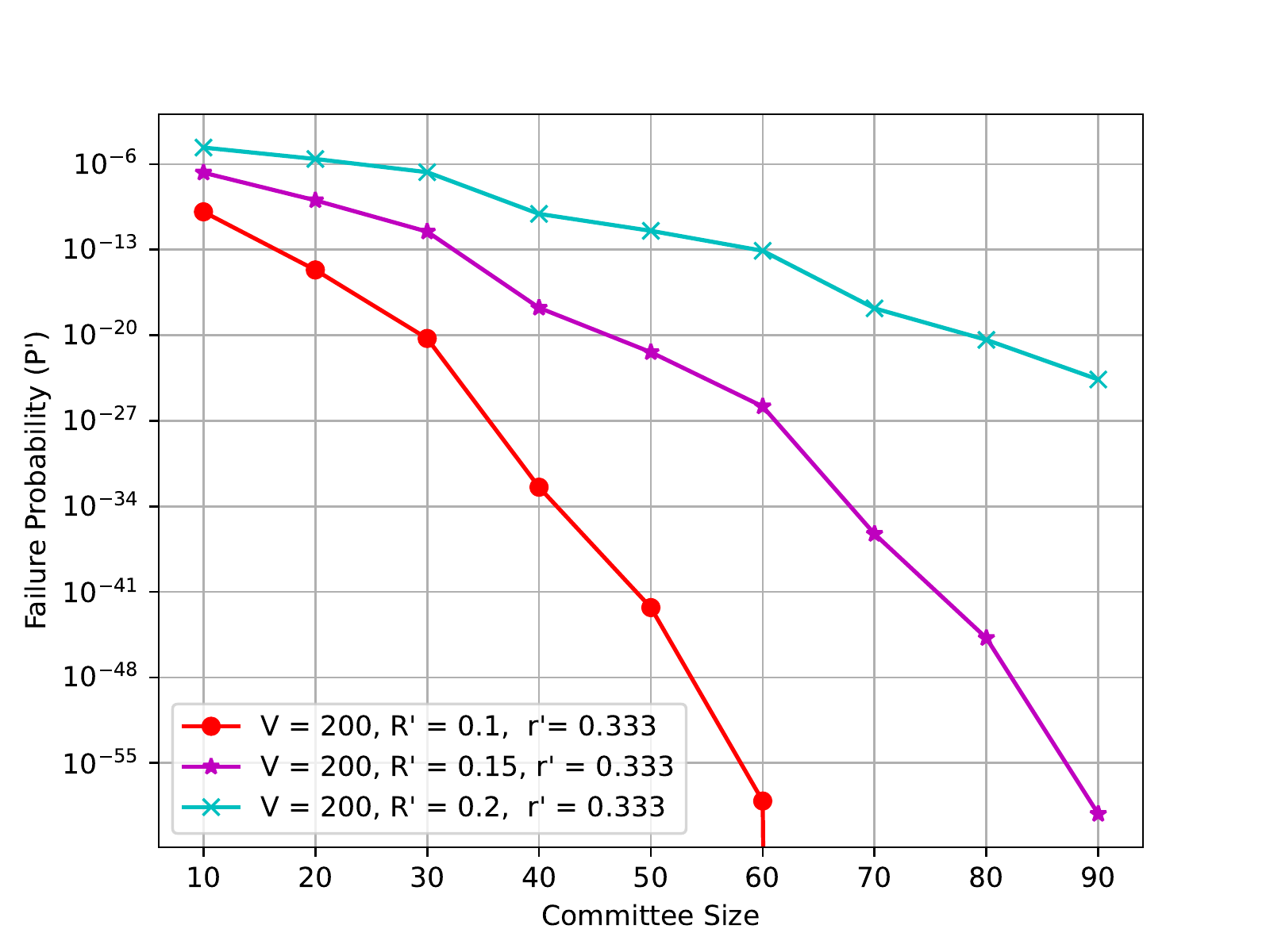}
\caption{Log-scale plot of the probability of the beacon chain to commit a faulty block ($P^{''}$) versus the size of the committee (n)}
\label{fig: Probability of all shards committing a faulty block}
\end{figure}

Figure \ref{fig: Probability of the beacon chain to commit a faulty block} shows the probability of the beacon chain to commit a faulty block versus the size of the committee. We observe that the probability $P^{''}$ decreases when the size of the committee increases. More specifically, we observe that the probability corresponding to $R=0.1$ (i.e. 10\% of malicious nodes in each shards) decreases sharply compared to those of $R=0.15$ and $R=0.2$.  


\begin{table}[ht]
\centering
\caption{Probability of conquering the protocol}
\setlength{\tabcolsep}{10pt}
\begin{tabular}{|c|c|c|c|c|} 
\hline
$p_m$ &
10 \% &
15 \% &
20 \% & 
30 \% \\
\hline
$\mathcal{P}^{\mathrm{a}}$ &
3.63E-66 &
2.10E-34 &
1.58E-18 & 
1.70E-04 \\
\hline
$Y_f^{\mathrm{a}}$ &
7.56E+62 &
1.30E+31 &
1.74E+17 & 
16.12 \\
\hline
$\mathcal{P}^{\mathrm{b}}$ &
 0.0 &
5.14E-80 &
2.01E-41 & 
5.30E-07 \\
\hline
$Y_f^{\mathrm{b}}$ &
inf &
5.33E+76 &
1.36E+38 & 
5171.32\\
\hline
\multicolumn{4}{l}{$^{\mathrm{a}}$Scenario 1; $^{\mathrm{b}}$Scenario 2.}
\end{tabular}
\label{tab: Probability of conquering the chain}
\end{table}

In Figure \ref{tab: Probability of conquering the chain}, we assume two scenarios to show the effectiveness and the feasibility of the proposed model: Scenario 1 proposes a network with $N = 2000$,  $\zeta = 8$, $V = 200$, $V'= 400$, and $r = r' = 0.333$ whereas Scenario 2 proposes a network with $N = 4000$,  $\zeta = 8$, $V = 400$, $V'= 800$, and $r = r' = 0.333$. It is noteworthy that the proposed model can be adopted to any scenario.

Table \ref{tab: Probability of conquering the chain} shows the probability of conquering the chain (i.e. the probability of committing a faulty block; it is calculated based on Theorem \ref{theorem 1}) for different percentage of malicious nodes in the shards as well as in the beacon chain. Moreover, Table \ref{tab: Probability of conquering the chain} shows the number of years to fail corresponding to these probabilities. More specifically, we observe that as the percentage of malicious nodes increases the number of years to fail decreases.

We observe that the probability of conquering the chain is extremely low even with 20\% of malicious nodes in each shard as well as in the beacon chain. This achieves a good security, which is about $1.74E+17$ years to fail.

To sum up, we conclude that by adjusting the size of the shard's committee as well as the size of the beacon's committee, we could support/protect the sharded Blockchain systems against malicious nodes (e.g. Sybil nodes).

\section{Conclusion} \label{sec: Conclusion}
This article addresses the security of sharding-based blockchain protocols that are based on PoS and pBFT consensuses. In particular, it provides a probabilistic model to compute the probability of committing a faulty block. Based on this probability, we compute the number of years to fail. Furthermore, this article depicts that we can control the number of years to fail by adjusting the size of the shard's committee as well as the size of the beacon's committee. Future works focus on computing the failure probability across-shard transaction.

\bibliographystyle{IEEEtran} 
\bibliography{Sharding}

\newpage

\end{document}